\documentclass[11pt]{article}

\usepackage[a4paper,margin=1in]{geometry}

\usepackage[T1]{fontenc}
\usepackage[utf8]{inputenc}
\usepackage{lmodern}
\usepackage{abstract} 
\usepackage{array}
\usepackage{amsmath}
\usepackage{amssymb}
\usepackage{amsthm}
\usepackage{mathtools}
\usepackage{mathrsfs}

\usepackage{enumitem}
\usepackage{booktabs}
\usepackage{multirow}
\usepackage{float}
\usepackage{pifont}
\usepackage{xcolor}

\usepackage{url}
\usepackage{hyperref}

\hypersetup{
	colorlinks=true,
	linkcolor=blue,
	citecolor=blue,
	urlcolor=blue
}

\newcommand{\Prop}{\mathsf{Prop}}       

\newcommand{\cmark}{\ensuremath{\checkmark}}
\newcommand{\xmark}{\ensuremath{\times}}

\newcommand{\CL}{\mathsf{CL}}           
\newcommand{\ATL}{\mathsf{ATL}}         
\newcommand{\CLIab}{\mathsf{CL}^{\Iab}} 

\newcommand{\Eff}[1]{\langle #1\rangle}

\newcommand{\Iab}{\mathsf{Iab}}

\theoremstyle{definition}
\newtheorem{definition}{Definition}[section]

\theoremstyle{plain}
\newtheorem{theorem}[definition]{Theorem}
\newtheorem{lemma}[definition]{Lemma}
\newtheorem{proposition}[definition]{Proposition}
\newtheorem{corollary}[definition]{Corollary}

\theoremstyle{remark}
\newtheorem{remark}[definition]{Remark}

\title{\textbf{A Logic of Inability}}

\author{
	Shanxia Wang\\
	School of Computer and Information Engineering (School of Artificial Intelligence)\\
	Henan Normal University, Xinxiang City, Henan Province, China\\
	\texttt{wangshanxia@htu.edu.cn}
}

\newenvironment{keywords}
{\par\noindent\textbf{Keywords: }}
{\par\medskip}

\date{}

\begin{document}
	
	\maketitle
	
	\begin{abstract}
		Coalition Logic is primarily concerned with what coalitions can achieve, whereas what coalitions cannot achieve---their \emph{inability}---has received comparatively little explicit attention. 
		This asymmetry matters in artificial intelligence and safety-critical multi-agent systems, where one often needs to specify not merely what agents are instructed or disposed not to do, but what they are \emph{unable} to bring about. 
		We develop a conservative extension of Coalition Logic with an explicit inability operator, interpreted as the negation of coalition ability. 
		This operator is introduced as a conservative and formally tractable starting point for studying inability as a modal concept in its own right. 
		We prove soundness, completeness, and conservativity over standard Coalition Logic, and analyse the resulting modal behaviour: anti-monotonicity with respect to coalition inclusion, contravariance with respect to goal strength, asymmetric interaction with conjunction and disjunction, failure of superadditivity, non-equivalence with opponent ability, and the connection between grand-coalition inability and systemic impossibility. 
		Making this definable operator explicit reveals a systematic modal structure governing the limits of agency, and supports reasoning about constraints, negative capabilities, and impossibility in multi-agent systems.
	\end{abstract}
	

	\section{Introduction}
	\label{sec:introduction}
	
	\subsection{Motivation: From Ability to Inability}
	\label{subsec:motivation}
	
	Inability is a pervasive concept in reasoning about agency.
	In everyday life, we distinguish what agents \emph{cannot} do from what they merely \emph{do not} do: a locked door renders a person unable to enter, not merely unwilling.
	In philosophy, the analysis of ability and its privation runs from Aristotle's \emph{dynamis} to contemporary debates on the conditional analysis of ``can''~\cite{Austin56,Kenny75,Maier22}.
	In ethics and law, inability carries normative force: the principle \emph{ought implies can}~\cite{Vranas07} entails that inability to comply with a norm extinguishes the obligation to comply.
	In artificial intelligence and multi-agent systems, inability is the substance of safety guarantees: a containment mechanism is trustworthy precisely because the contained system is \emph{unable} to circumvent it~\cite{Bostrom14,Amodei16,Ji23,BH24}; a voting rule is strategy-proof precisely because every voter is \emph{unable} to profit from misrepresentation~\cite{Gibbard73,Satterthwaite75}.
	
	Across these domains, inability functions as a concept in its own right---not merely as the absence of ability, but as a substantive condition that grounds excuses, certifies safety, and delimits responsibility.
	
	Yet when we turn to formal logic, and to the logic of coalitional agency in particular, we find a striking asymmetry.
	The focus has been overwhelmingly on \emph{ability}: what agents and coalitions \emph{can} achieve.
	Inability appears only derivatively, as negated ability.
	This paper initiates the systematic study of inability as an independent object of logical investigation.
	
	\subsection{The Neglect of Inability in Coalition Logic}
	\label{subsec:neglect}
	
	Coalition Logic~\cite{Pauly02,Pauly01} provides operators \(\Eff{C}\varphi\), read ``coalition \(C\) can ensure \(\varphi\).''
	Together with its temporal extension \(\ATL\)~\cite{AHK02}, it has become a standard framework for reasoning about strategic power in multi-agent systems~\cite{Wooldridge09,BGJ15}.
	Subsequent work has enriched the ability modality with epistemic conditions~\cite{HW03,AA19,AJ22}, resource bounds~\cite{ADL14,ABL17}, action constraints~\cite{GLP24,LH22}, and strategy quantifiers~\cite{MMPV14,CLMR23}.
	The result is a mature and productive theory of what coalitions \emph{can} do.
	
	No comparable theory exists for what coalitions \emph{cannot} do.
	
	In standard Coalition Logic, inability is expressed by negation:
	\[
	\neg\Eff{C}\varphi.
	\]
	It has no dedicated operator, no independent axiomatisation, and no systematic study of its structural properties.
	It lives entirely in the shadow of ability.
	
	This is technically sufficient but conceptually limiting.
	The situation is analogous to studying necessity \(\Box\varphi\) without ever isolating possibility \(\Diamond\varphi\), or studying knowledge \(K_i\varphi\)~\cite{Hintikka62} without ever examining ignorance.
	In each case the negative notion is \emph{definable} from the positive one.
	But definability does not imply dispensability: a concept may deserve explicit treatment when it organises reasoning in a distinctive way, reveals structural patterns that are otherwise obscured, and opens questions that do not naturally arise from the positive concept alone.
	
	\subsection{Goal of This Paper: Inability as a First-Class Concept}
	\label{subsec:goal}
	
	The goal of this paper is to establish inability as a first-class concept in the logic of coalitional agency, on equal footing with ability.
	
	We work within classical Coalition Logic and introduce an explicit inability operator \(\Iab_C\varphi\), read ``coalition \(C\) is unable to ensure \(\varphi\),'' defined by
	\[
	\Iab_C\varphi \;\;\leftrightarrow\;\; \neg\Eff{C}\varphi.
	\]
	Coalition \(C\) is unable to ensure \(\varphi\) when it has no joint action that guarantees \(\varphi\) regardless of the actions of the remaining agents.
	The resulting system \(\CLIab\) is a conservative extension of Coalition Logic: every theorem of \(\CLIab\) in the original language is already a theorem of \(\CL\).
	
	This conservativity is by design.
	Our aim is not to increase expressive power but to \emph{reorient} the logical perspective---to place inability at the centre of attention and to develop its theory systematically.
	With the operator \(\Iab_C\) in hand, we investigate:
	\begin{itemize}[nosep]
		\item How does inability behave under coalition expansion and contraction?
		\item How does it interact with goal strength, conjunction, and disjunction?
		\item What is the relationship between a coalition's inability and the ability of its opponents?
	\end{itemize}
	
	The answers reveal a coherent modal profile whose main features are summarised in the contributions below.
	
	\subsection{A Starting Point, Not an Endpoint}
	\label{subsec:starting-point}
	
	We study one specific form of inability: the strategic inability of a coalition to guarantee an outcome in a one-shot concurrent-game setting.
	This is the simplest notion of inability available in Coalition Logic, and we do not claim that it exhausts the concept.
	We view this work as a \emph{starting point} for the logical study of inability, not an endpoint.
	
	The precedent from epistemic logic is instructive.
	Ignorance began as mere negation of knowledge, \(\neg K_i\varphi\), but its explicit study generated a rich theory with its own distinctions and results~\cite{HM85,FH88,vdHL04,BFG23,FH24}.
	Unawareness was shown to collapse in standard Kripke semantics~\cite{DLR98}, yet this very impossibility result spurred the development of new semantic frameworks~\cite{Schipper15,DFP24,HR14} in which unawareness could be modelled non-trivially.
	In both cases, a concept that was initially treated as a simple negation eventually acquired an independent logical life.
	
	We expect inability to follow a similar trajectory.
	Strategic inability under imperfect information~\cite{HW03,JLN23}, resource-bounded inability~\cite{ADL14,ABL17}, institutional inability, and temporal inability in \(\ATL\)-like settings~\cite{AHK02,GHL22} are natural next steps.
	The present paper lays the groundwork: it singles out inability as an explicit concept, identifies its basic laws, and provides the formal foundation on which richer theories can be built.
	
	\subsection{Contributions}
	\label{subsec:contributions}
	
	\begin{enumerate}
		\item \textbf{Inability operator and semantics.}
		We introduce \(\Iab_C\varphi\) as \(\neg\Eff{C}\varphi\) and give its game-theoretic semantics within coalition models, making inability a first-class modality.
		
		\item \textbf{Axiomatisation.}
		We define the logic \(\CLIab\), prove soundness and completeness with respect to coalition models, and establish conservativity over~\(\CL\).
		
		\item \textbf{Structural laws.}
		We develop the modal theory of inability: anti-monotonicity with respect to coalition inclusion, contravariance with respect to goal strength, asymmetric distribution over Boolean connectives, failure of superadditivity, non-equivalence with opponent ability, and exact dualities at boundary coalitions.
		
		\item \textbf{Conceptual positioning.}
		We situate inability among related negative modal concepts, argue for the value of explicit treatment despite definability, and outline applications to AI safety, social choice, and protocol verification.
	\end{enumerate}
	
	\subsection{Organisation}
	\label{subsec:organisation}
	
	Section~\ref{sec:preliminaries} recalls Coalition Logic.
	Section~\ref{sec:inability} introduces the inability operator, develops \(\CLIab\), and proves soundness, completeness, and conservativity.
	Section~\ref{sec:structural-laws} presents the structural laws of inability.
	Section~\ref{sec:discussion} discusses conceptual significance, methodological precedents, and applications.
	Section~\ref{sec:conclusion} concludes.

	\section{Preliminaries: Coalition Logic}
	\label{sec:preliminaries}
	
	This section recalls the syntax, semantics, and axiomatics of Coalition Logic~\cite{Pauly02,Pauly01}.
	Coalition Logic is a modal logic for reasoning about the strategic powers of groups of agents.
	Its central modality
	\[
	\Eff{C}\varphi
	\]
	expresses that coalition \(C\) can ensure \(\varphi\), no matter how the agents outside \(C\) act.
	The presentation below uses a one-step concurrent-game semantics, corresponding to the standard game-frame interpretation of Coalition Logic.
	This is the setting in which the inability operator introduced in the next section will be interpreted.
	
	\subsection{Syntax and Semantics}
	
	Let \(N = \{1,\ldots,n\}\) be a finite non-empty set of agents.
	A \emph{coalition} is any subset \(C \subseteq N\).
	The complement of \(C\) is denoted by \(\overline{C} = N \setminus C\).
	Let \(\Prop\) be a countable set of propositional variables.
	
	\begin{definition}[Language]
		\label{def:cl-language}
		The language \(\mathcal{L}_{\CL}\) of Coalition Logic is defined by:
		\[
		\varphi ::= p \mid \neg\varphi \mid (\varphi \wedge \psi) \mid \Eff{C}\varphi,
		\]
		where \(p \in \Prop\) and \(C \subseteq N\).
	\end{definition}
	
	Boolean connectives \(\vee,\rightarrow,\leftrightarrow\) and constants \(\top,\bot\) are defined as usual.
	The formula \(\Eff{C}\varphi\) is read as ``coalition \(C\) can ensure \(\varphi\).''
	
	The semantics is given by coalition models, which formalise one-step strategic interaction among agents.
	Each agent independently chooses an action, and the combination of all agents' choices determines the outcome state.
	
	\begin{definition}[Coalition model]
		\label{def:coalition-model}
		A \emph{coalition model} is a tuple
		\[
		\mathcal{M} = (S,\{Act_i\}_{i\in N},o,V),
		\]
		where:
		\begin{itemize}
			\item \(S\) is a non-empty set of states;
			\item \(Act_i\) is a non-empty set of actions for each agent \(i\in N\);
			\item \(o:S\times \prod_{i\in N} Act_i \to S\) is an outcome function;
			\item \(V:\Prop \to 2^S\) is a valuation.
		\end{itemize}
	\end{definition}
	
	For a coalition \(C\subseteq N\), let
	\[
	Act_C=\prod_{i\in C}Act_i
	\]
	be the set of joint actions of \(C\).
	For \(C=\emptyset\), \(Act_C\) contains the unique empty action profile.
	Given \(\sigma_C\in Act_C\) and \(\sigma_{\overline{C}}\in Act_{\overline{C}}\), we write
	\[
	o(s,\sigma_C,\sigma_{\overline{C}})
	\]
	for the outcome state when coalition \(C\) plays \(\sigma_C\) and the complementary coalition \(\overline{C}\) plays \(\sigma_{\overline{C}}\).
	In this one-step setting, the term ``strategy'' is occasionally used informally for such a joint action; no temporal or history-dependent strategy is involved.
	
	\begin{definition}[Satisfaction]
		\label{def:cl-satisfaction}
		The satisfaction relation \(\mathcal{M},s\models\varphi\) is defined recursively:
		\begin{align*}
			\mathcal{M},s &\models p
			&&\text{iff} \quad s\in V(p), \\
			\mathcal{M},s &\models \neg\varphi
			&&\text{iff} \quad \mathcal{M},s\not\models \varphi, \\
			\mathcal{M},s &\models \varphi\wedge\psi
			&&\text{iff} \quad \mathcal{M},s\models\varphi \text{ and } \mathcal{M},s\models\psi, \\
			\mathcal{M},s &\models \Eff{C}\varphi
			&&\text{iff} \quad
			\exists \sigma_C\in Act_C\;
			\forall \sigma_{\overline{C}}\in Act_{\overline{C}}:
			\mathcal{M},o(s,\sigma_C,\sigma_{\overline{C}})\models\varphi.
		\end{align*}
	\end{definition}
	
	The modal clause captures the game-theoretic interpretation: \(\Eff{C}\varphi\) holds at state \(s\) when coalition \(C\) has a joint action \(\sigma_C\) such that, no matter how the complementary coalition \(\overline{C}\) acts, the resulting state satisfies \(\varphi\).
	Thus Coalition Logic describes one-step enforceability rather than temporal achievement.
	Temporal extensions such as \(\ATL\)~\cite{AHK02} add explicit operators for reasoning about strategies over paths; see also~\cite{GHL22} for recent developments in game-theoretic semantics for such extensions.
	
	A formula \(\varphi\) is \emph{valid}, written \(\models\varphi\), if \(\mathcal{M},s\models\varphi\) for every coalition model \(\mathcal{M}\) and every state \(s\) in \(\mathcal{M}\).
	
	\begin{remark}
		Coalition models of this kind correspond to game frames in the sense of Pauly~\cite{Pauly02}.
		An alternative semantic approach uses effectivity functions~\cite{Pauly01}, which abstract away from explicit actions and record only which sets of outcomes each coalition can enforce.
		For the one-step language considered here, game-frame semantics and the corresponding playable effectivity-function semantics validate the same set of principles; see~\cite{GJ04} for a detailed comparison.
	\end{remark}
	
	\subsection{Axiomatics}
	
	A standard axiom system for \(\CL\) provides a complete proof theory for Coalition Logic.
	We use the following presentation.
	
	\begin{definition}[The system \(\CL\)]
		\label{def:cl-system}
		The axiom system \(\CL\) consists of the following axiom schemes and rules.
		
		\medskip
		\noindent\textbf{Axiom schemes:}
		\begin{description}
			\item[(PL)] All propositional tautologies.
			
			\item[(T)] \(\Eff{C}\top\), for every \(C\subseteq N\).
			\hfill Every coalition can ensure truth.
			
			\item[(M)] \(\neg\Eff{C}\bot\), for every \(C\subseteq N\).
			\hfill No coalition can ensure contradiction.
			
			\item[(S)]
			\[
			(\Eff{C}\varphi \wedge \Eff{D}\psi)
			\rightarrow
			\Eff{C\cup D}(\varphi\wedge\psi),
			\]
			if \(C\cap D=\emptyset\).
			\hfill Superadditivity.
			
			\item[(G)]
			\[
			\neg\Eff{\emptyset}\neg\varphi \rightarrow \Eff{N}\varphi.
			\]
			\hfill Grand coalition power.
		\end{description}
		
		\medskip
		\noindent\textbf{Rules:}
		\begin{description}
			\item[(MP)] From \(\varphi\) and \(\varphi\rightarrow\psi\), infer \(\psi\).
			\hfill Modus ponens.
			
			\item[(RE)] From \(\vdash \varphi\leftrightarrow\psi\), infer
			\[
			\vdash \Eff{C}\varphi \leftrightarrow \Eff{C}\psi.
			\]
			\hfill Replacement of provable equivalents.
		\end{description}
	\end{definition}
	
	Axiom (T) says that every coalition can ensure truth.
	Together with (RE), this yields that every coalition can ensure any formula provably equivalent to \(\top\).
	Axiom (M) states that no coalition can ensure a contradiction.
	Axiom (S) captures the superadditivity of coalition power: if two disjoint coalitions can separately ensure \(\varphi\) and \(\psi\), then their union can ensure the conjunction \(\varphi\wedge\psi\).
	Axiom (G) expresses the power of the grand coalition: if the empty coalition cannot ensure \(\neg\varphi\)---that is, if not every complete action profile leads to a \(\neg\varphi\)-state---then some complete action profile leads to a \(\varphi\)-state, and the grand coalition can choose such a profile.
	
	Coalition Logic is not a normal modal logic in the usual Kripkean sense~\cite{Chellas80,BRV01}.
	The coalition modality does not in general satisfy the distribution axiom of normal modal logic:
	\[
	\Eff{C}(\varphi\rightarrow\psi)
	\rightarrow
	(\Eff{C}\varphi\rightarrow\Eff{C}\psi).
	\]
	Intuitively, the joint action by which \(C\) ensures \(\varphi\rightarrow\psi\) need not be the same joint action by which \(C\) ensures \(\varphi\), so the two guarantees cannot in general be composed.
	Note, however, that Coalition Logic does satisfy the semantic analogue of necessitation: if \(\models\varphi\), then \(\models\Eff{C}\varphi\) for every coalition \(C\), since every possible outcome state satisfies a valid formula.
	The non-normality lies specifically in the failure of distribution.
	
	\begin{theorem}[Pauly~\cite{Pauly02}]
		\label{thm:cl-complete}
		The system \(\CL\) is sound and complete with respect to the class of coalition models: for every \(\varphi\in\mathcal{L}_{\CL}\),
		\[
		\CL\vdash\varphi
		\quad\text{iff}\quad
		\models\varphi.
		\]
	\end{theorem}
	
	We take this result as background.
	Since the inability operator introduced in the next section is definitionally reducible to coalition ability, the metatheory of the extended system will follow by translation into \(\CL\).
	
	\begin{remark}
		Coalition Logic admits extensions in several directions, including temporal operators~\cite{AHK02}, epistemic conditions~\cite{HW03,AA19}, resource constraints~\cite{ADL14}, and strategy quantifiers~\cite{MMPV14}, as surveyed in Section~\ref{sec:introduction}.
		The present paper focuses on basic one-step Coalition Logic, leaving such extensions to future work.
	\end{remark}

	\section{The Logic of Inability}
	\label{sec:inability}
	
	This section introduces the inability operator, provides its semantics, defines the extended system \(\CLIab\), and establishes soundness, completeness, and conservativity over standard Coalition Logic.
	The inability studied here is coalitional and strategic: \(\Iab_C\varphi\) means that coalition \(C\) has no one-step joint action that guarantees \(\varphi\) against all actions of the complementary coalition.
	
	\subsection{Syntax and Semantics}
	
	We extend the language of Coalition Logic with an explicit inability operator.
	
	\begin{definition}[Extended language]
		\label{def:iab-language}
		The language \(\mathcal{L}_{\Iab}\) is defined by:
		\[
		\varphi ::= p \mid \neg\varphi \mid (\varphi \wedge \psi) \mid \Eff{C}\varphi \mid \Iab_C\varphi,
		\]
		where \(p\in\Prop\) and \(C\subseteq N\).
	\end{definition}
	
	Boolean connectives \(\vee,\rightarrow,\leftrightarrow\) and constants \(\top,\bot\) are defined as usual.
	The formula \(\Iab_C\varphi\) is read as ``coalition \(C\) is unable to ensure \(\varphi\).''
	Its intended meaning is that there is no joint action available to \(C\) that guarantees \(\varphi\), no matter how the agents outside \(C\) act.
	
	The inability operator is governed by the following definitional axiom scheme.
	
	\begin{definition}[Definitional axiom for inability]
		\label{def:iab-def}
		For every coalition \(C\subseteq N\) and formula \(\varphi\):
		\[
		\textup{(Iab-Def)} \qquad
		\Iab_C\varphi \leftrightarrow \neg\Eff{C}\varphi.
		\]
	\end{definition}
	
	Thus \(C\) is unable to ensure \(\varphi\) if and only if it is not the case that \(C\) can ensure \(\varphi\).
	The operator \(\Iab_C\) is introduced as a definitional extension of the standard coalition modality.
	The point is not to increase expressive power, but to make explicit a modal notion whose structural behaviour will be studied in its own right.
	
	The semantic clause follows directly from this definitional principle.
	
	\begin{definition}[Semantics of inability]
		\label{def:iab-semantics}
		For a coalition model \(\mathcal{M}\), state \(s\), coalition \(C\subseteq N\), and formula \(\varphi\in\mathcal{L}_{\Iab}\),
		\[
		\mathcal{M},s \models \Iab_C\varphi
		\quad\text{iff}\quad
		\mathcal{M},s \not\models \Eff{C}\varphi.
		\]
		Equivalently,
		\[
		\mathcal{M},s \models \Iab_C\varphi
		\quad\text{iff}\quad
		\forall \sigma_C\in Act_C\;
		\exists \sigma_{\overline{C}}\in Act_{\overline{C}}:
		\mathcal{M},o(s,\sigma_C,\sigma_{\overline{C}})\not\models\varphi.
		\]
	\end{definition}
	
	This clause is the game-theoretic dual of coalition ability.
	While \(\Eff{C}\varphi\) says that \(C\) has a joint action that guarantees \(\varphi\) against all compatible actions of \(\overline{C}\), the formula \(\Iab_C\varphi\) says that every joint action available to \(C\) can be countered by some action of \(\overline{C}\) leading to a state where \(\varphi\) fails.
	
	It is important to read this carefully.
	The existentially quantified action of \(\overline{C}\) may depend on the action chosen by \(C\).
	Thus \(\Iab_C\varphi\) does not, by itself, say that \(\overline{C}\) has a single joint action that ensures \(\neg\varphi\).
	In particular, one should not identify
	\[
	\Iab_C\varphi
	\quad\text{with}\quad
	\Eff{\overline{C}}\neg\varphi.
	\]
	This distinction will be formally established in Section~\ref{sec:structural-laws}.
	It is closely related to the classical game-theoretic distinction between strategies and best responses in simultaneous-move games~\cite{OR94}.
	
	\begin{remark}[Inability versus falsity]
		\label{rem:iab-vs-falsity}
		The formula \(\Iab_C\varphi\) should not be confused with \(\neg\varphi\).
		Inability is a modal claim about strategic power, not a factual claim about the current state.
		It is possible that
		\[
		\mathcal{M},s\models \varphi \wedge \Iab_C\varphi.
		\]
		In such a case, \(\varphi\) happens to hold at \(s\), but coalition \(C\) cannot guarantee that \(\varphi\) will hold after a joint action profile is executed.
		Conversely, \(\Iab_C\varphi\) does not imply that \(\varphi\) is currently false; it only says that \(C\) lacks the power to ensure \(\varphi\).
	\end{remark}
	
	\begin{remark}[Inability versus actual failure]
		\label{rem:iab-vs-failure}
		The formula \(\Iab_C\varphi\) is also distinct from claims about what coalition \(C\) actually does.
		A coalition may be able to ensure \(\varphi\) but fail to choose an action that does so; conversely, a coalition may aim at \(\varphi\) without having any action that guarantees it.
		Inability concerns the absence of guaranteed enforceability, not the actual choice or outcome.
		This distinction between what agents \emph{can} bring about and what they \emph{do} bring about is central to the logic of agency~\cite{BPX01,Horty01,BH15}.
	\end{remark}
	
	\subsection{The System \texorpdfstring{\(\CLIab\)}{CLIab}}
	
	The logic \(\CLIab\) extends Coalition Logic with the definitional axiom for inability.
	
	\begin{definition}[The system \(\CLIab\)]
		\label{def:CLIab-system}
		The axiom system \(\CLIab\) consists of:
		\begin{itemize}
			\item all axiom schemes and rules of \(\CL\), now applied to formulas of \(\mathcal{L}_{\Iab}\);
			\item the axiom scheme
			\[
			\textup{(Iab-Def)} \qquad
			\Iab_C\varphi \leftrightarrow \neg\Eff{C}\varphi.
			\]
		\end{itemize}
	\end{definition}
	
	Since \(\Iab_C\varphi\) is definitionally equivalent to \(\neg\Eff{C}\varphi\), every formula of \(\mathcal{L}_{\Iab}\) can be translated into a formula of the original language \(\mathcal{L}_{\CL}\).
	
	\begin{definition}[Translation]
		\label{def:iab-translation}
		The translation
		\[
		t:\mathcal{L}_{\Iab}\to\mathcal{L}_{\CL}
		\]
		is defined recursively as follows:
		\begin{align*}
			t(p) &= p, \\
			t(\neg\varphi) &= \neg t(\varphi), \\
			t(\varphi\wedge\psi) &= t(\varphi)\wedge t(\psi), \\
			t(\Eff{C}\varphi) &= \Eff{C}t(\varphi), \\
			t(\Iab_C\varphi) &= \neg\Eff{C}t(\varphi).
		\end{align*}
	\end{definition}
	
	The translation eliminates every occurrence of the inability operator by replacing it with its defining equivalent.
	It is compositional and preserves the Boolean and modal structure of formulae.
	
	\subsection{Metatheory}
	
	We now show that \(\CLIab\) is a conservative definitional extension of \(\CL\) in the sense of~\cite{Enderton01}.
	Soundness and completeness follow by translation into the original language of Coalition Logic.
	
	\begin{lemma}[Truth preservation]
		\label{lem:truth-preservation}
		For every \(\varphi\in\mathcal{L}_{\Iab}\), every coalition model \(\mathcal{M}\), and every state \(s\),
		\[
		\mathcal{M},s\models\varphi
		\quad\text{iff}\quad
		\mathcal{M},s\models t(\varphi).
		\]
	\end{lemma}
	
	\begin{proof}
		By induction on the structure of \(\varphi\).
		The propositional and Boolean cases are straightforward.
		
		For the ability case \(\varphi=\Eff{C}\psi\), we have \(t(\Eff{C}\psi)=\Eff{C}t(\psi)\).
		By the induction hypothesis, for every state \(s'\),
		\[
		\mathcal{M},s'\models\psi
		\quad\text{iff}\quad
		\mathcal{M},s'\models t(\psi).
		\]
		Therefore, for every \(\sigma_C\in Act_C\) and \(\sigma_{\overline{C}}\in Act_{\overline{C}}\),
		\[
		\mathcal{M},o(s,\sigma_C,\sigma_{\overline{C}})\models\psi
		\quad\text{iff}\quad
		\mathcal{M},o(s,\sigma_C,\sigma_{\overline{C}})\models t(\psi).
		\]
		Hence
		\[
		\mathcal{M},s\models\Eff{C}\psi
		\quad\text{iff}\quad
		\mathcal{M},s\models\Eff{C}t(\psi) = t(\Eff{C}\psi).
		\]
		
		For the inability case \(\varphi=\Iab_C\psi\), by Definition~\ref{def:iab-semantics},
		\[
		\mathcal{M},s\models \Iab_C\psi
		\quad\text{iff}\quad
		\mathcal{M},s\not\models \Eff{C}\psi.
		\]
		By the ability case,
		\[
		\mathcal{M},s\models \Eff{C}\psi
		\quad\text{iff}\quad
		\mathcal{M},s\models \Eff{C}t(\psi).
		\]
		Thus
		\[
		\mathcal{M},s\models \Iab_C\psi
		\quad\text{iff}\quad
		\mathcal{M},s\models \neg\Eff{C}t(\psi)
		= t(\Iab_C\psi). \qedhere
		\]
	\end{proof}
	
	\begin{lemma}[Provable equivalence]
		\label{lem:provable-equivalence}
		For every \(\varphi\in\mathcal{L}_{\Iab}\),
		\[
		\CLIab\vdash \varphi \leftrightarrow t(\varphi).
		\]
	\end{lemma}
	
	\begin{proof}
		By induction on the structure of \(\varphi\).
		The propositional and Boolean cases are straightforward.
		
		For the ability case \(\varphi=\Eff{C}\psi\), the induction hypothesis gives \(\CLIab\vdash \psi \leftrightarrow t(\psi)\).
		By rule (RE),
		\[
		\CLIab\vdash \Eff{C}\psi \leftrightarrow \Eff{C}t(\psi)
		=
		t(\Eff{C}\psi).
		\]
		
		For the inability case \(\varphi=\Iab_C\psi\), the induction hypothesis and (RE) yield \(\CLIab\vdash \Eff{C}\psi \leftrightarrow \Eff{C}t(\psi)\), whence by propositional reasoning \(\CLIab\vdash \neg\Eff{C}\psi \leftrightarrow \neg\Eff{C}t(\psi)\).
		By (Iab-Def),
		\[
		\CLIab\vdash \Iab_C\psi \leftrightarrow \neg\Eff{C}\psi
		\leftrightarrow \neg\Eff{C}t(\psi)
		=
		t(\Iab_C\psi). \qedhere
		\]
	\end{proof}
	
	\begin{theorem}[Soundness and completeness]
		\label{thm:CLIab-complete}
		For every \(\varphi\in\mathcal{L}_{\Iab}\),
		\[
		\CLIab\vdash\varphi
		\quad\text{iff}\quad
		\models\varphi.
		\]
	\end{theorem}
	
	\begin{proof}
		\emph{Soundness.}
		The axiom schemes inherited from \(\CL\) remain valid when applied to formulae of \(\mathcal{L}_{\Iab}\).
		The definitional axiom (Iab-Def) is valid by Definition~\ref{def:iab-semantics}.
		The rules (MP) and (RE) preserve validity.
		Therefore every theorem of \(\CLIab\) is valid.
		
		\medskip
		
		\emph{Completeness.}
		Suppose \(\models\varphi\).
		By Lemma~\ref{lem:truth-preservation}, \(\models t(\varphi)\).
		Since \(t(\varphi)\in\mathcal{L}_{\CL}\), by Theorem~\ref{thm:cl-complete}, \(\CL\vdash t(\varphi)\).
		Since every theorem of \(\CL\) is a theorem of \(\CLIab\), we have \(\CLIab\vdash t(\varphi)\).
		By Lemma~\ref{lem:provable-equivalence}, \(\CLIab\vdash \varphi \leftrightarrow t(\varphi)\).
		Hence, by modus ponens, \(\CLIab\vdash \varphi\).
	\end{proof}
	
	\begin{theorem}[Conservative extension]
		\label{thm:conservative}
		The system \(\CLIab\) is a conservative extension of \(\CL\): for every \(\varphi\in\mathcal{L}_{\CL}\),
		\[
		\CLIab\vdash\varphi
		\quad\text{iff}\quad
		\CL\vdash\varphi.
		\]
	\end{theorem}
	
	\begin{proof}
		The right-to-left direction is immediate, since \(\CLIab\) includes all axioms and rules of \(\CL\).
		For the left-to-right direction, suppose \(\varphi\in\mathcal{L}_{\CL}\) and \(\CLIab\vdash\varphi\).
		By Theorem~\ref{thm:CLIab-complete}, \(\models\varphi\).
		Since \(\varphi\in\mathcal{L}_{\CL}\), by Theorem~\ref{thm:cl-complete}, \(\CL\vdash\varphi\).
	\end{proof}
	
	\begin{remark}[Why make inability explicit?]
		\label{rem:first-class}
		The conservativity result confirms that \(\Iab_C\) adds no expressive power to Coalition Logic.
		This is intentional.
		The aim of the extension is to isolate inability as a modal notion whose structural behaviour can be stated and studied directly, rather than to obtain new definability results.
		Working only with formulae of the form \(\neg\Eff{C}\varphi\) is technically sufficient, but it tends to obscure the patterns that emerge once inability is treated as an explicit modality.
		The next section develops these patterns systematically.
	\end{remark}

	\section{Structural Laws of Inability}
	\label{sec:structural-laws}
	
	Although inability is introduced as the Boolean dual of coalition ability,
	\[
	\Iab_C\varphi \leftrightarrow \neg\Eff{C}\varphi,
	\]
	making it explicit reveals a distinctive structural profile.
	The operator \(\Iab_C\) obeys laws governing coalition inclusion, goal strength, Boolean connectives, and boundary coalitions.
	Some of these laws are direct contrapositives of familiar ability principles; others mark important limits on how inability should be understood.
	In particular, inability of \(C\) to ensure \(\varphi\) should not be confused with ability of the complementary coalition \(\overline{C}\) to ensure \(\neg\varphi\).
	
	All validity and invalidity claims in this section are understood over the class of coalition models defined in Section~\ref{sec:preliminaries}.
	All countermodels below are evaluated at a designated initial state \(s\).
	Outcomes at states other than \(s\), and any outcomes not explicitly specified, may be chosen arbitrarily, since they play no role in the relevant evaluations.
	
	\subsection{Anti-Monotonicity over Coalitions}
	
	Coalition ability is monotonic with respect to coalition inclusion: larger coalitions inherit the capabilities of smaller ones~\cite{Pauly02}.
	Inability exhibits the corresponding dual behaviour.
	
	\begin{proposition}[Monotonicity of ability]
		\label{prop:ability-mono}
		For all coalitions \(C\subseteq D\subseteq N\),
		\[
		\models \Eff{C}\varphi \rightarrow \Eff{D}\varphi.
		\]
	\end{proposition}
	
	\begin{proof}
		Suppose \(\mathcal{M},s\models\Eff{C}\varphi\).
		Then there exists \(\sigma_C\in Act_C\) such that for every \(\sigma_{\overline{C}}\in Act_{\overline{C}}\),
		\[
		\mathcal{M},o(s,\sigma_C,\sigma_{\overline{C}})\models\varphi.
		\]
		Since \(C\subseteq D\), choose any \(\tau_{D\setminus C}\in Act_{D\setminus C}\), and let \(\sigma_D=(\sigma_C,\tau_{D\setminus C})\in Act_D\).
		For every \(\sigma_{\overline{D}}\in Act_{\overline{D}}\), the pair \((\tau_{D\setminus C},\sigma_{\overline{D}})\) forms a joint action of \(\overline{C}\).
		By the choice of \(\sigma_C\), the corresponding outcome satisfies \(\varphi\).
		Hence \(\mathcal{M},s\models\Eff{D}\varphi\).
	\end{proof}
	
	\begin{theorem}[Anti-monotonicity of inability]
		\label{thm:inability-antimono}
		For all coalitions \(C\subseteq D\subseteq N\),
		\[
		\models \Iab_D\varphi \rightarrow \Iab_C\varphi.
		\]
	\end{theorem}
	
	\begin{proof}
		By Proposition~\ref{prop:ability-mono} and propositional contraposition,
		\[
		\models \neg\Eff{D}\varphi \rightarrow \neg\Eff{C}\varphi.
		\]
		By (Iab-Def), this is exactly \(\models \Iab_D\varphi \rightarrow \Iab_C\varphi\).
	\end{proof}
	
	Anti-monotonicity captures the intuition that if a larger coalition cannot ensure a goal, then no smaller subcoalition can ensure it either.
	Additional agents may add possible coordinated actions; they do not remove the abilities already available to a smaller coalition.
	The converse direction fails in general.
	
	\begin{proposition}[Failure of upward propagation]
		\label{prop:upward-failure}
		There exist coalitions \(C\subsetneq D\) such that
		\[
		\not\models \Iab_C\varphi \rightarrow \Iab_D\varphi.
		\]
	\end{proposition}
	
	\begin{proof}
		Let \(N=\{1,2\}\), \(C=\{1\}\), \(D=\{1,2\}\), \(Act_1=Act_2=\{a,b\}\), \(S=\{s,t,u\}\), and \(V(p)=\{t\}\).
		Define \(o(s,x,y)=t\) iff \(x=a\) and \(y=a\), and let all other outcomes at \(s\) be \(u\).
		
		Agent \(1\) alone cannot ensure \(p\): if agent \(1\) plays \(a\), agent \(2\) may play \(b\); if agent \(1\) plays \(b\), the outcome is \(u\).
		Thus \(\mathcal{M},s\models \Iab_{\{1\}}p\).
		However, the grand coalition \(\{1,2\}\) can ensure \(p\) by choosing \((a,a)\).
		Hence \(\mathcal{M},s\not\models \Iab_{\{1,2\}}p\).
	\end{proof}
	
	\subsection{Contravariance over Goal Strength}
	
	Ability is covariant with respect to logical strength: if \(\varphi\) entails \(\psi\), then ensuring \(\varphi\) entails ensuring \(\psi\).
	Inability is contravariant.
	
	\begin{proposition}[Covariance of ability]
		\label{prop:ability-covariance}
		If \(\models \varphi \rightarrow \psi\), then
		\[
		\models \Eff{C}\varphi \rightarrow \Eff{C}\psi.
		\]
	\end{proposition}
	
	\begin{proof}
		Suppose \(\mathcal{M},s\models\Eff{C}\varphi\).
		Then \(C\) has a joint action such that every resulting state satisfies \(\varphi\).
		Since \(\models\varphi\rightarrow\psi\), every such state also satisfies \(\psi\).
		Therefore \(\mathcal{M},s\models\Eff{C}\psi\).
	\end{proof}
	
	\begin{theorem}[Contravariance of inability]
		\label{thm:inability-contravariance}
		If \(\models \varphi \rightarrow \psi\), then
		\[
		\models \Iab_C\psi \rightarrow \Iab_C\varphi.
		\]
	\end{theorem}
	
	\begin{proof}
		By Proposition~\ref{prop:ability-covariance} and propositional contraposition,
		\[
		\models \neg\Eff{C}\psi \rightarrow \neg\Eff{C}\varphi.
		\]
		By (Iab-Def), this is \(\models \Iab_C\psi \rightarrow \Iab_C\varphi\).
	\end{proof}
	
	Contravariance reflects the fact that weaker goals are easier to ensure.
	If \(C\) cannot ensure a weaker goal \(\psi\), then \(C\) cannot ensure any stronger goal \(\varphi\) with \(\models\varphi\rightarrow\psi\).
	The converse direction fails.
	
	\begin{proposition}[Failure of covariance for inability]
		\label{prop:inability-covariance-failure}
		There exist formulae \(\varphi,\psi\) such that \(\models\varphi\rightarrow\psi\), but
		\[
		\not\models \Iab_C\varphi \rightarrow \Iab_C\psi.
		\]
	\end{proposition}
	
	\begin{proof}
		Let \(N=\{1,2\}\), \(C=\{1\}\), \(Act_1=\{a\}\), \(Act_2=\{b_1,b_2\}\), \(S=\{s,t_1,t_2\}\), \(V(p)=\{t_1,t_2\}\), and \(V(q)=\{t_1\}\).
		Define \(o(s,a,b_1)=t_1\) and \(o(s,a,b_2)=t_2\).
		Let \(\varphi=p\wedge q\) and \(\psi=p\).
		Clearly, \(\models (p\wedge q)\rightarrow p\).
		
		At \(s\), coalition \(C=\{1\}\) can ensure \(p\), since its only action \(a\) leads to either \(t_1\) or \(t_2\), both satisfying \(p\).
		Thus \(\mathcal{M},s\not\models \Iab_{\{1\}}p\).
		However, \(C\) cannot ensure \(p\wedge q\), since agent \(2\) may choose \(b_2\), leading to \(t_2\) where \(q\) is false.
		Hence \(\mathcal{M},s\models \Iab_{\{1\}}(p\wedge q)\).
	\end{proof}
	
	\begin{remark}[Direction of propagation]
		\label{rem:direction-of-propagation}
		The direction of propagation is easily mistaken.
		Inability to ensure a strong goal says nothing about inability to ensure a weaker goal.
		The valid principle is the opposite: inability to ensure the weaker goal propagates to inability to ensure the stronger one.
	\end{remark}
	
	\subsection{Interaction with Boolean Connectives}
	
	\subsubsection*{Conjunction}
	
	Inability propagates from conjuncts to conjunctions, but not conversely.
	
	\begin{theorem}[Downward distribution over conjunction]
		\label{thm:conj-downward}
		\[
		\models (\Iab_C\varphi \vee \Iab_C\psi) \rightarrow \Iab_C(\varphi\wedge\psi).
		\]
	\end{theorem}
	
	\begin{proof}
		By Theorem~\ref{thm:inability-contravariance} applied to \(\models (\varphi\wedge\psi)\rightarrow\varphi\) and \(\models (\varphi\wedge\psi)\rightarrow\psi\), we obtain
		\[
		\models \Iab_C\varphi \rightarrow \Iab_C(\varphi\wedge\psi)
		\quad\text{and}\quad
		\models \Iab_C\psi \rightarrow \Iab_C(\varphi\wedge\psi).
		\]
		By propositional reasoning, \(\models (\Iab_C\varphi \vee \Iab_C\psi) \rightarrow \Iab_C(\varphi\wedge\psi)\).
	\end{proof}
	
	\begin{corollary}[Absorption]
		\label{cor:absorption}
		\[
		\models \Iab_C\varphi \rightarrow \Iab_C(\varphi\wedge\psi).
		\]
	\end{corollary}
	
	\begin{proof}
		This is the left disjunct case of Theorem~\ref{thm:conj-downward}.
	\end{proof}
	
	\begin{proposition}[Failure of upward distribution over conjunction]
		\label{prop:conj-upward-failure}
		\[
		\not\models \Iab_C(\varphi\wedge\psi) \rightarrow (\Iab_C\varphi \vee \Iab_C\psi).
		\]
	\end{proposition}
	
	\begin{proof}
		Let \(N=\{1\}\), \(C=\{1\}\), \(Act_1=\{a,b\}\), \(S=\{s,t_1,t_2\}\), \(V(p)=\{t_1\}\), and \(V(q)=\{t_2\}\).
		Define \(o(s,a)=t_1\) and \(o(s,b)=t_2\).
		
		Agent \(1\) can ensure \(p\) by choosing \(a\), and can ensure \(q\) by choosing \(b\).
		Hence \(\mathcal{M},s\not\models \Iab_{\{1\}}p\) and \(\mathcal{M},s\not\models \Iab_{\{1\}}q\).
		However, no action leads to a state satisfying both \(p\) and \(q\).
		Thus \(\mathcal{M},s\models \Iab_{\{1\}}(p\wedge q)\).
	\end{proof}
	
	\subsubsection*{Disjunction}
	
	For disjunction, the valid direction is the reverse: inability to ensure a disjunction entails inability to ensure either disjunct, but inability to ensure each disjunct separately need not entail inability to ensure their disjunction.
	
	\begin{theorem}[Upward distribution over disjunction]
		\label{thm:disj-upward}
		\[
		\models \Iab_C(\varphi\vee\psi) \rightarrow (\Iab_C\varphi \wedge \Iab_C\psi).
		\]
	\end{theorem}
	
	\begin{proof}
		By Theorem~\ref{thm:inability-contravariance} applied to \(\models \varphi\rightarrow(\varphi\vee\psi)\) and \(\models \psi\rightarrow(\varphi\vee\psi)\), we obtain
		\[
		\models \Iab_C(\varphi\vee\psi)\rightarrow\Iab_C\varphi
		\quad\text{and}\quad
		\models \Iab_C(\varphi\vee\psi)\rightarrow\Iab_C\psi.
		\]
		Therefore, \(\models \Iab_C(\varphi\vee\psi) \rightarrow (\Iab_C\varphi \wedge \Iab_C\psi)\).
	\end{proof}
	
	\begin{proposition}[Failure of downward distribution over disjunction]
		\label{prop:disj-downward-failure}
		\[
		\not\models (\Iab_C\varphi \wedge \Iab_C\psi) \rightarrow \Iab_C(\varphi\vee\psi).
		\]
	\end{proposition}
	
	\begin{proof}
		Let \(N=\{1,2\}\), \(C=\{1\}\), \(Act_1=\{a\}\), \(Act_2=\{b_1,b_2\}\), \(S=\{s,t_1,t_2\}\), \(V(p)=\{t_1\}\), and \(V(q)=\{t_2\}\).
		Define \(o(s,a,b_1)=t_1\) and \(o(s,a,b_2)=t_2\).
		
		Coalition \(C=\{1\}\) cannot ensure \(p\), since agent \(2\) may choose \(b_2\).
		Likewise, \(C\) cannot ensure \(q\), since agent \(2\) may choose \(b_1\).
		Thus \(\mathcal{M},s\models \Iab_{\{1\}}p \wedge \Iab_{\{1\}}q\).
		However, every possible outcome after \(C\)'s only action satisfies \(p\vee q\).
		Therefore \(\mathcal{M},s\models \Eff{\{1\}}(p\vee q)\), and hence \(\mathcal{M},s\not\models \Iab_{\{1\}}(p\vee q)\).
	\end{proof}
	
	\begin{remark}[Coarse-grained versus fine-grained control]
		\label{rem:coarse-fine-control}
		A coalition may be able to guarantee that some member of a set of desirable outcomes occurs without being able to determine which member occurs.
		This is coarse-grained control.
		It suffices for disjunctive goals, but not for the ability to ensure either disjunct individually.
		This phenomenon is well known in the study of effectivity functions~\cite{Pauly01} and has analogues in the distinction between knowledge \emph{de dicto} and \emph{de re} in epistemic logic~\cite{FHMV95}.
	\end{remark}
	
	\subsubsection*{Implication}
	
	Implication can be treated through its Boolean equivalence with disjunction.
	
	\begin{theorem}[Distribution into implication]
		\label{thm:implication-distribution}
		\[
		\models \Iab_C(\varphi\rightarrow\psi) \rightarrow (\Iab_C\neg\varphi \wedge \Iab_C\psi).
		\]
	\end{theorem}
	
	\begin{proof}
		Since \(\varphi\rightarrow\psi\) is propositionally equivalent to \(\neg\varphi\vee\psi\), by Theorem~\ref{thm:disj-upward},
		\[
		\models \Iab_C(\neg\varphi\vee\psi)
		\rightarrow
		(\Iab_C\neg\varphi \wedge \Iab_C\psi).
		\]
		By replacement of provable equivalents, \(\models \Iab_C(\varphi\rightarrow\psi) \rightarrow (\Iab_C\neg\varphi \wedge \Iab_C\psi)\).
	\end{proof}
	
	\begin{proposition}[Failure of the converse implication principle]
		\label{prop:implication-converse-failure}
		\[
		\not\models (\Iab_C\neg\varphi \wedge \Iab_C\psi) \rightarrow \Iab_C(\varphi\rightarrow\psi).
		\]
	\end{proposition}
	
	\begin{proof}
		Use the model from Proposition~\ref{prop:disj-downward-failure}.
		Instantiate \(\varphi\) as \(\neg p\) and \(\psi\) as \(q\).
		Then \(\neg\varphi\) is \(p\), and \(\varphi\rightarrow\psi\) is \(\neg p\rightarrow q\), which is propositionally equivalent to \(p\vee q\).
		The antecedent \(\Iab_C\neg\varphi \wedge \Iab_C\psi\) holds at \(s\), while the consequent \(\Iab_C(\varphi\rightarrow\psi)\) fails at \(s\).
	\end{proof}
	
	\subsection{Coalition Composition}
	
	Ability is superadditive: disjoint coalitions that can ensure separate goals can jointly ensure their conjunction (axiom (S) of Definition~\ref{def:cl-system}).
	Inability does not satisfy an analogous superadditivity principle.
	The valid general principle is instead a subadditivity law.
	
	\begin{proposition}[Superadditivity of ability]
		\label{prop:superadditivity}
		For disjoint coalitions \(C,D\),
		\[
		\models (\Eff{C}\varphi \wedge \Eff{D}\psi) \rightarrow \Eff{C\cup D}(\varphi\wedge\psi).
		\]
	\end{proposition}
	
	\begin{proof}
		Suppose \(C\cap D=\emptyset\), \(C\) has a joint action \(\sigma_C\) ensuring \(\varphi\), and \(D\) has a joint action \(\sigma_D\) ensuring \(\psi\).
		Since \(C\) and \(D\) are disjoint, the two joint actions combine into a joint action \((\sigma_C,\sigma_D)\) of \(C\cup D\).
		For any action of \(\overline{C\cup D}\), the resulting complete action profile extends both \(\sigma_C\) and \(\sigma_D\), so the outcome satisfies both \(\varphi\) and \(\psi\).
	\end{proof}
	
	\begin{theorem}[Failure of superadditivity for inability]
		\label{thm:superadd-failure}
		The following superadditivity principle for inability is not valid in general, even when \(C\cap D=\emptyset\):
		\[
		(\Iab_C\varphi \wedge \Iab_D\psi)
		\rightarrow
		\Iab_{C\cup D}(\varphi\wedge\psi).
		\]
	\end{theorem}
	
	\begin{proof}
		Let \(N=\{1,2\}\), \(C=\{1\}\), \(D=\{2\}\), \(Act_1=Act_2=\{a,b\}\), \(S=\{s,t,u\}\), and \(V(p)=\{t\}\).
		Define \(o(s,x,y)=t\) iff \(x=a\) and \(y=a\), and let all other outcomes at \(s\) be \(u\).
		
		Agent \(1\) alone cannot ensure \(p\), since agent \(2\) may choose \(b\).
		Agent \(2\) alone cannot ensure \(p\), since agent \(1\) may choose \(b\).
		Thus \(\mathcal{M},s\models \Iab_{\{1\}}p \wedge \Iab_{\{2\}}p\).
		However, the coalition \(\{1,2\}\) can ensure \(p\) by choosing \((a,a)\).
		Therefore \(\mathcal{M},s\not\models \Iab_{\{1,2\}}p\).
		Taking \(\varphi=\psi=p\), the proposed superadditivity principle for inability is invalid.
	\end{proof}
	
	\begin{theorem}[Subadditivity of inability]
		\label{thm:subadditivity}
		\[
		\models \Iab_{C\cup D}\varphi \rightarrow (\Iab_C\varphi \wedge \Iab_D\varphi).
		\]
	\end{theorem}
	
	\begin{proof}
		Since \(C\subseteq C\cup D\) and \(D\subseteq C\cup D\), by anti-monotonicity of inability (Theorem~\ref{thm:inability-antimono}),
		\[
		\models \Iab_{C\cup D}\varphi \rightarrow \Iab_C\varphi
		\quad\text{and}\quad
		\models \Iab_{C\cup D}\varphi \rightarrow \Iab_D\varphi.
		\]
		Combining yields \(\models \Iab_{C\cup D}\varphi \rightarrow (\Iab_C\varphi \wedge \Iab_D\varphi)\).
	\end{proof}
	
	\begin{remark}[Cooperation and the dissolution of inability]
		\label{rem:cooperation-dissolves-inability}
		The failure of superadditivity for inability captures a central feature of collective agency: limitations of individual coalitions may be overcome by cooperation~\cite{Wooldridge09}.
		The valid direction is subadditivity: if the union of two coalitions is unable to ensure a goal, then each part is unable to ensure it on its own.
	\end{remark}
	
	\subsection{Excluded Middle and Strategic Choice}
	
	The next principles concern the relation between inability, choice, and non-zero-sum interaction.
	
	\begin{theorem}[Failure of excluded middle for inability]
		\label{thm:excluded-middle-failure}
		\[
		\not\models \Iab_C\varphi \vee \Iab_C\neg\varphi.
		\]
	\end{theorem}
	
	\begin{proof}
		Let \(N=\{1\}\), \(C=\{1\}\), \(Act_1=\{a,b\}\), \(S=\{s,t_1,t_2\}\), and \(V(p)=\{t_1\}\).
		Define \(o(s,a)=t_1\) and \(o(s,b)=t_2\).
		
		Agent \(1\) can ensure \(p\) by choosing \(a\), and can ensure \(\neg p\) by choosing \(b\).
		Hence \(\mathcal{M},s\not\models \Iab_{\{1\}}p\) and \(\mathcal{M},s\not\models \Iab_{\{1\}}\neg p\).
		Therefore \(\mathcal{M},s\not\models \Iab_{\{1\}}p \vee \Iab_{\{1\}}\neg p\).
	\end{proof}
	
	\begin{theorem}[Failure of exclusivity]
		\label{thm:exclusivity-failure}
		\[
		\not\models \Eff{C}\varphi \rightarrow \Iab_{\overline{C}}\varphi.
		\]
	\end{theorem}
	
	\begin{proof}
		Take \(\varphi=\top\).
		Every coalition can ensure \(\top\), since \(\top\) holds at every state.
		Thus \(\models \Eff{C}\top\) and \(\models \Eff{\overline{C}}\top\).
		By (Iab-Def), \(\models \neg\Iab_{\overline{C}}\top\).
		Hence \(\Eff{C}\top \rightarrow \Iab_{\overline{C}}\top\) is not valid.
	\end{proof}
	
	\begin{remark}[Strategic freedom]
		\label{rem:strategic-freedom}
		The failure of excluded middle reflects the possibility of genuine strategic control: a coalition may be able to ensure \(\varphi\) and also able to ensure \(\neg\varphi\), by choosing different actions.
		The failure of exclusivity reflects the non-zero-sum character of Coalition Logic: the fact that one coalition can ensure a goal does not imply that the complementary coalition is unable to ensure the same goal.
		This non-zero-sum feature distinguishes Coalition Logic from strictly competitive game models~\cite{OR94}.
	\end{remark}
	
	\subsection{Symmetry and Complementarity}
	
	A tempting but misleading thought is that inability of \(C\) to ensure \(\varphi\) should correspond to inability of \(\overline{C}\) to ensure \(\neg\varphi\), or to ability of \(\overline{C}\) to ensure \(\neg\varphi\).
	Neither correspondence is valid in general.
	
	\begin{theorem}[Failure of symmetry]
		\label{thm:symmetry-failure}
		\[
		\not\models \Iab_C\varphi \leftrightarrow \Iab_{\overline{C}}\neg\varphi.
		\]
	\end{theorem}
	
	\begin{proof}
		Let \(N=\{1,2\}\), \(C=\{1\}\), \(Act_1=\{a,b\}\), \(Act_2=\{c\}\), \(S=\{s,t,u\}\), and \(V(p)=\{t\}\).
		Define \(o(s,a,c)=t\) and \(o(s,b,c)=u\).
		
		Coalition \(C=\{1\}\) can ensure \(p\) by choosing \(a\).
		Therefore \(\mathcal{M},s\not\models \Iab_{\{1\}}p\).
		The complementary coalition \(\overline{C}=\{2\}\) has only one action \(c\), and agent \(1\) may choose \(a\), leading to \(t\) where \(\neg p\) is false.
		Thus \(\{2\}\) cannot ensure \(\neg p\), and so \(\mathcal{M},s\models \Iab_{\{2\}}\neg p\).
		Hence \(\mathcal{M},s\not\models \Iab_{\{1\}}p \leftrightarrow \Iab_{\{2\}}\neg p\).
	\end{proof}
	
	\begin{theorem}[Failure of complementarity]
		\label{thm:complementarity-failure}
		\[
		\not\models \Iab_C\varphi \vee \Iab_{\overline{C}}\varphi.
		\]
	\end{theorem}
	
	\begin{proof}
		Take \(\varphi=\top\).
		Every coalition can ensure \(\top\), including both \(C\) and \(\overline{C}\).
		Thus, by (Iab-Def), \(\models \neg\Iab_C\top\) and \(\models \neg\Iab_{\overline{C}}\top\).
		Therefore \(\Iab_C\top \vee \Iab_{\overline{C}}\top\) is not valid.
	\end{proof}
	
	\begin{theorem}[Inability does not imply opponent ability]
		\label{thm:inability-opponent}
		\[
		\not\models \Iab_C\varphi \rightarrow \Eff{\overline{C}}\neg\varphi.
		\]
	\end{theorem}
	
	\begin{proof}
		Consider a matching-pennies model~\cite{OR94}.
		Let \(N=\{1,2\}\), \(C=\{1\}\), \(Act_1=Act_2=\{H,T\}\), \(S=\{s,t,u\}\), and \(V(p)=\{t\}\).
		Define \(o(s,x,y)=t\) iff \(x=y\), and let \(o(s,x,y)=u\) otherwise.
		
		Agent \(1\) cannot ensure \(p\): whatever agent \(1\) chooses, agent \(2\) can choose the opposite action, producing \(u\) where \(p\) is false.
		Thus \(\mathcal{M},s\models \Iab_{\{1\}}p\).
		However, agent \(2\) cannot ensure \(\neg p\): whatever agent \(2\) chooses, agent \(1\) can choose the same action, producing \(t\) where \(p\) is true.
		Hence \(\mathcal{M},s\not\models \Eff{\{2\}}\neg p\).
	\end{proof}
	
	\begin{remark}[Strategic indeterminacy]
		\label{rem:strategic-indeterminacy}
		Theorem~\ref{thm:inability-opponent} is central for interpreting inability.
		The statement \(\Iab_C\varphi\) says that \(C\) lacks a guaranteed way to ensure \(\varphi\).
		It does not say that \(\overline{C}\) has a guaranteed way to ensure \(\neg\varphi\).
		In simultaneous-move settings, each action of \(C\) may have some counteraction by \(\overline{C}\), without there being a single action of \(\overline{C}\) that defeats all actions of \(C\).
		This phenomenon is familiar from the theory of two-person zero-sum games without a value in pure strategies~\cite{OR94}.
	\end{remark}
	
	\subsection{Boundary Coalitions}
	
	The empty coalition and the grand coalition produce exact dualities.
	These boundary cases are useful for distinguishing coalitional inability from systemic impossibility and contingency.
	
	\begin{theorem}[Grand coalition inability]
		\label{thm:grand-coalition}
		\[
		\models \Iab_N\varphi \leftrightarrow \Eff{\emptyset}\neg\varphi.
		\]
	\end{theorem}
	
	\begin{proof}
		At a state \(s\), \(\mathcal{M},s\models \Iab_N\varphi\) iff \(\mathcal{M},s\not\models \Eff{N}\varphi\).
		Since \(N\) is the grand coalition, its complement is empty.
		Thus \(\Eff{N}\varphi\) holds at \(s\) iff there exists a complete action profile leading to a \(\varphi\)-state.
		Therefore \(\Iab_N\varphi\) holds iff every complete action profile leads to a state satisfying \(\neg\varphi\).
		
		This is exactly the condition for \(\mathcal{M},s\models \Eff{\emptyset}\neg\varphi\), because the empty coalition has only the empty joint action, and \(\Eff{\emptyset}\neg\varphi\) requires that all complete action profiles lead to \(\neg\varphi\)-states.
		Hence \(\models \Iab_N\varphi \leftrightarrow \Eff{\emptyset}\neg\varphi\).
	\end{proof}
	
	\begin{theorem}[Empty coalition inability]
		\label{thm:empty-coalition}
		\[
		\models \Iab_{\emptyset}\varphi \leftrightarrow \Eff{N}\neg\varphi.
		\]
	\end{theorem}
	
	\begin{proof}
		At a state \(s\), \(\mathcal{M},s\models \Iab_{\emptyset}\varphi\) iff \(\mathcal{M},s\not\models \Eff{\emptyset}\varphi\).
		The formula \(\Eff{\emptyset}\varphi\) holds iff every complete action profile leads to a \(\varphi\)-state.
		Therefore \(\Iab_{\emptyset}\varphi\) holds iff there exists some complete action profile leading to a state satisfying \(\neg\varphi\).
		
		Since the grand coalition \(N\) controls a complete action profile, this is exactly the condition for \(\mathcal{M},s\models \Eff{N}\neg\varphi\).
		Hence \(\models \Iab_{\emptyset}\varphi \leftrightarrow \Eff{N}\neg\varphi\).
	\end{proof}
	
	\begin{remark}[Duality at the boundary]
		\label{rem:boundary-duality}
		The boundary coalitions yield exact dualities:
		\[
		\Iab_N\varphi \leftrightarrow \Eff{\emptyset}\neg\varphi,
		\qquad
		\Iab_{\emptyset}\varphi \leftrightarrow \Eff{N}\neg\varphi.
		\]
		Grand coalition inability corresponds to the absence of any complete action profile achieving \(\varphi\): the goal is systemically unachievable.
		Empty coalition inability corresponds to the existence of some complete action profile leading to \(\neg\varphi\): the goal is contingent rather than settled.
		These dualities are the coalition-logic analogues of the classical modal duality between necessity and possibility~\cite{BRV01}.
	\end{remark}
	
	\subsection{Logical Boundary Cases}
	
	Finally, inability behaves predictably on logical constants.
	
	\begin{proposition}[Logical boundary cases]
		\label{prop:boundary}
		For every coalition \(C\subseteq N\):
		\begin{enumerate}
			\item \(\models \Iab_C\bot\).
			\item \(\models \neg\Iab_C\top\).
		\end{enumerate}
	\end{proposition}
	
	\begin{proof}
		For (1), no coalition can ensure contradiction: \(\models \neg\Eff{C}\bot\).
		By (Iab-Def), this is equivalent to \(\models \Iab_C\bot\).
		
		For (2), every coalition can ensure truth.
		Indeed, since each \(Act_C\) is non-empty and \(\top\) holds at every state, any joint action of \(C\) ensures \(\top\).
		Thus \(\models \Eff{C}\top\).
		By (Iab-Def), \(\models \neg\Iab_C\top\).
	\end{proof}
	
	\subsection{Summary}
	
	Table~\ref{tab:structural-laws} summarises the structural laws established in this section.
	
	\begin{table}[H]
		\centering
		\caption{Structural laws of inability: \(\cmark\) = valid, \(\xmark\) = invalid.}
		\label{tab:structural-laws}
		\footnotesize
		\renewcommand{\arraystretch}{1.2}
		\begin{tabular}{@{}lp{5.8cm}c@{}}
			\toprule
			\textbf{Principle} & \textbf{Formula} & \textbf{Status} \\
			\midrule
			\multicolumn{3}{l}{\textit{Coalition properties}} \\
			Anti-monotonicity
			& \(\Iab_D\varphi \rightarrow \Iab_C\varphi\) \quad (\(C\subseteq D\))
			& \(\cmark\) \\
			Upward propagation
			& \(\Iab_C\varphi \rightarrow \Iab_D\varphi\) \quad (\(C\subsetneq D\))
			& \(\xmark\) \\
			Subadditivity
			& \(\Iab_{C\cup D}\varphi \rightarrow (\Iab_C\varphi\wedge\Iab_D\varphi)\)
			& \(\cmark\) \\
			Superadditivity
			& \((\Iab_C\varphi\wedge\Iab_D\psi) \rightarrow \Iab_{C\cup D}(\varphi\wedge\psi)\) \quad (\(C\cap D=\emptyset\))
			& \(\xmark\) \\
			\midrule
			
			\multicolumn{3}{l}{\textit{Goal properties}} \\
			Contravariance
			& \(\Iab_C\psi \rightarrow \Iab_C\varphi\) \quad if \(\models\varphi\rightarrow\psi\)
			& \(\cmark\) \\
			Covariance
			& \(\Iab_C\varphi \rightarrow \Iab_C\psi\) \quad if \(\models\varphi\rightarrow\psi\)
			& \(\xmark\) \\
			Absorption
			& \(\Iab_C\varphi \rightarrow \Iab_C(\varphi\wedge\psi)\)
			& \(\cmark\) \\
			\midrule
			
			\multicolumn{3}{l}{\textit{Boolean connectives}} \\
			Conjunction, downward
			& \((\Iab_C\varphi\vee\Iab_C\psi) \rightarrow \Iab_C(\varphi\wedge\psi)\)
			& \(\cmark\) \\
			Conjunction, upward
			& \(\Iab_C(\varphi\wedge\psi) \rightarrow (\Iab_C\varphi\vee\Iab_C\psi)\)
			& \(\xmark\) \\
			Disjunction, upward
			& \(\Iab_C(\varphi\vee\psi) \rightarrow (\Iab_C\varphi\wedge\Iab_C\psi)\)
			& \(\cmark\) \\
			Disjunction, downward
			& \((\Iab_C\varphi\wedge\Iab_C\psi) \rightarrow \Iab_C(\varphi\vee\psi)\)
			& \(\xmark\) \\
			Implication
			& \(\Iab_C(\varphi\rightarrow\psi) \rightarrow (\Iab_C\neg\varphi\wedge\Iab_C\psi)\)
			& \(\cmark\) \\
			Implication converse
			& \((\Iab_C\neg\varphi\wedge\Iab_C\psi) \rightarrow \Iab_C(\varphi\rightarrow\psi)\)
			& \(\xmark\) \\
			\midrule
			
			\multicolumn{3}{l}{\textit{Strategic properties}} \\
			Excluded middle
			& \(\Iab_C\varphi \vee \Iab_C\neg\varphi\)
			& \(\xmark\) \\
			Exclusivity
			& \(\Eff{C}\varphi \rightarrow \Iab_{\overline{C}}\varphi\)
			& \(\xmark\) \\
			Symmetry
			& \(\Iab_C\varphi \leftrightarrow \Iab_{\overline{C}}\neg\varphi\)
			& \(\xmark\) \\
			Complementarity
			& \(\Iab_C\varphi \vee \Iab_{\overline{C}}\varphi\)
			& \(\xmark\) \\
			Opponent ability
			& \(\Iab_C\varphi \rightarrow \Eff{\overline{C}}\neg\varphi\)
			& \(\xmark\) \\
			\midrule
			
			\multicolumn{3}{l}{\textit{Boundary coalitions and constants}} \\
			Grand coalition
			& \(\Iab_N\varphi \leftrightarrow \Eff{\emptyset}\neg\varphi\)
			& \(\cmark\) \\
			Empty coalition
			& \(\Iab_{\emptyset}\varphi \leftrightarrow \Eff{N}\neg\varphi\)
			& \(\cmark\) \\
			Contradiction
			& \(\Iab_C\bot\)
			& \(\cmark\) \\
			Truth
			& \(\neg\Iab_C\top\)
			& \(\cmark\) \\
			\bottomrule
		\end{tabular}
	\end{table}
	
	The valid principles---anti-monotonicity, contravariance, subadditivity, the asymmetric distribution laws, and the boundary dualities---provide a toolkit for reasoning about strategic limitations.
	The invalid principles---especially the failures of upward propagation, superadditivity, excluded middle, symmetry, complementarity, and opponent ability---show that inability cannot be treated as a simple transfer of power to the complementary coalition.
	Together, these laws constitute a coherent structural profile that justifies treating \(\Iab_C\) as an explicit modality.

	\section{Discussion}
	\label{sec:discussion}
	
	This section discusses the conceptual significance of making inability explicit in Coalition Logic and identifies the scope and limits of the present approach.
	
	\subsection{Definability and Conceptual Significance}
	
	As noted in Section~\ref{sec:introduction}, a natural objection to \(\CLIab\) is that the inability operator adds no expressive power: every formula containing \(\Iab_C\) can be translated back into the original language of Coalition Logic.
	This objection is correct but not decisive.
	Definability does not entail dispensability, as the standard examples of \(\Diamond\) and the existential quantifier illustrate~\cite{BRV01,Enderton01}.
	
	The structural results of Section~\ref{sec:structural-laws} provide concrete evidence for this claim.
	Once inability is treated explicitly, a coherent profile emerges: inability is anti-monotonic with respect to coalition inclusion, contravariant with respect to goal strength, asymmetric in its interaction with conjunction and disjunction, and not equivalent to the opponent's ability to force failure.
	These facts are all expressible in the original language, but without \(\Iab_C\) they appear only as scattered consequences of principles governing ability.
	With \(\Iab_C\), they form a systematic theory of strategic limitation.
	
	The formula \(\neg\Eff{C}\varphi\) states inability indirectly, by denying an ability claim.
	The formula \(\Iab_C\varphi\) states it directly: \(\varphi\) lies outside the enforceable range of coalition \(C\).
	This shift in focal structure makes limitations and constraints available as explicit objects of formal reasoning, while leaving their definitional dependence on the standard ability modality fully transparent.
	
	\subsection{Inability, Ignorance, and Unawareness}
	
	The treatment of inability has a useful methodological precedent in epistemic logic.
	Ignorance and unawareness are often introduced through negative or dependent notions: ignorance as lack of knowledge, unawareness as absence of awareness~\cite{FHMV95,HR14}.
	Yet both have motivated extensive independent study because they capture epistemic conditions that are not always well understood merely as the negation of positive attitudes.
	
	The analogy is best understood at the methodological level, not at the level of formal structure.
	The inability operator studied here expresses lack of ability, not a two-sided condition analogous to ``not knowing whether.''
	The stronger condition
	\[
	\neg\Eff{C}\varphi \wedge \neg\Eff{C}\neg\varphi
	\]
	expresses a stronger form of inability: coalition $C$ can ensure 
	neither $\varphi$ nor its negation, and is thus powerless with 
	respect to the truth value of $\varphi$ entirely. 
	While $\Iab_{C}\varphi$ captures the inability of a coalition to 
	bring about the truth of $\varphi$, this compound condition captures 
	the inability to bring about either its truth or its falsity---what 
	we may call \emph{strategic indeterminacy}. 
	A systematic study of this and other graded forms of inability 
	is left to future work.

	The closer analogy is with the methodological role of unawareness.
	Even when a limiting notion is initially introduced through its relation to a positive operator, its explicit study may reveal enough conceptual structure to justify richer semantic treatments later.
	In the present paper, inability is introduced conservatively within standard Coalition Logic without an independent semantics.
	Nevertheless, the explicit study identifies a stable structural pattern and clarifies where stronger notions may later enter: ability relative to information~\cite{HW03,AJ22}, awareness, resources~\cite{ADL14,ABL17}, institutions, or dynamic change.
	
	Thus the point of $\Iab_{C}$ is to treat inability 
	as a first-class modality, thereby revealing principles 
	that remain obscured when inability is merely negated ability.

	\subsection{Inability as a First-Class Modality}
	
	In multi-agent reasoning, many important specifications are
	naturally negative~\cite{Wooldridge09}.
	One often wants to say not merely that a desirable state is
	reachable, but that an undesirable state is not enforceable
	by some agent or coalition.
	Treating inability as a first-class modality means giving
	this negative dimension its own primitive operator
	$\Iab_C\varphi$, with its own axioms and inference rules,
	rather than encoding it as $\neg\Eff{C}\varphi$ and
	relying entirely on the theory of $\Eff{C}$.
	
	The formula $\Iab_C\varphi$ states that coalition $C$ lacks
	the strategic power to guarantee outcome $\varphi$ in the
	current one-step interaction.
	This should be understood strategically rather than
	deontically: $\Iab_C\varphi$ does not say that $C$ is
	\emph{forbidden} to bring about $\varphi$, nor that
	$\varphi$ is impossible in an absolute sense.
	The distinction between strategic inability and deontic
	prohibition is important: the former concerns what agents
	\emph{can} do, while the latter concerns what they
	\emph{may} do~\cite{BPX01,Horty01}.
	
	The distinction between inability and opponent ability,
	established in Section~\ref{sec:structural-laws}, is
	especially relevant for applications.
	To say that $C$ is unable to ensure $\varphi$ is not to say
	that the complementary coalition can ensure $\neg\varphi$.
	This matters in simultaneous-move or non-zero-sum
	interaction~\cite{OR94}, where neither side may have
	decisive control.
	Inability is a notion of lack of guaranteed control, not a
	transfer of control to an opponent.
	
	As outlined in Section~\ref{sec:introduction}, the language
	of inability matches the form of central claims in several
	domains.
	In AI safety, many requirements concern what a system cannot
	force~\cite{Bostrom14,Amodei16}.
	The formula $\Iab_H\mathsf{harm}$ states that the system
	lacks the power to guarantee a harmful outcome---a direct
	characterisation of the system's enforceable boundary,
	stronger than saying that harm does not currently occur, and
	different from saying that some other coalition can prevent
	harm.
	It should not, however, be read as a complete safety
	guarantee: harm may still occur under some joint action
	profiles.
	
	In protocol verification, security properties often take the
	form of negative capabilities~\cite{CLMR23}: an adversarial
	coalition should be unable to bring about key compromise,
	agent impersonation, or protocol inconsistency.
	In social choice and institutional design, impossibility
	results can be viewed as claims of systemic
	inability~\cite{Arrow51,Sen70,Gibbard73,Satterthwaite75,BCG16}:
	certain outcomes lie beyond the enforceable range of the
	institutional system under the given constraints.
	
	These examples are not meant as full reductions of AI
	safety, protocol analysis, or social choice theory to
	Coalition Logic.
	Rather, they show that treating inability as a first-class
	modality provides a direct way to express bounded agency,
	negative capability, and systemic impossibility---claims
	whose logical structure is more naturally captured by the
	primitive $\Iab_C\varphi$ than by the derived expression
	$\neg\Eff{C}\varphi$.

	\subsection{Scope and Limits}
	
	The present framework studies the basic form of inability available in Coalition Logic.
	It is one-step, outcome-oriented, and conservative over the standard ability modality.
	This gives the operator a clean technical foundation, but it also marks the limits of the present account.
	
	In particular, the logic developed here does not distinguish objective inability from lack of knowledge of ability~\cite{HW03,AA19}, unawareness of available actions, resource-bounded inability~\cite{ADL14,ABL17}, institutional prohibition, or computational infeasibility.
	Nor does it model the temporal evolution of inability, such as becoming able through learning, coordination, or resource acquisition.
	These distinctions require additional semantic structure, such as that provided by temporal extensions like \(\ATL\)~\cite{AHK02} or dynamic logics of games~\cite{GHL22,vDHK07}.
	
	The contribution of the present paper is therefore foundational.
	It shows that even the conservative notion of inability has a nontrivial and systematic logical structure.
	By making inability explicit, we obtain a language better suited to reasoning about limits of agency, negative requirements, and structural impossibilities in multi-agent systems, without departing from the standard semantic foundations of Coalition Logic.

	\section{Conclusion}
	\label{sec:conclusion}
	
	This paper has made inability an explicit object of study in Coalition Logic.
	Starting from the conservative definition
	\[
	\Iab_C\varphi \leftrightarrow \neg\Eff{C}\varphi,
	\]
	we showed that the resulting system \(\CLIab\) is sound, complete, and conservative over \(\CL\).
	The inability operator adds no expressive power, but it makes explicit a modal perspective that is otherwise hidden: the perspective of strategic limitation.
	
	We established the main structural laws of inability.
	Inability is anti-monotonic with respect to coalition inclusion and contravariant with respect to goal strength.
	Its interaction with Boolean connectives is asymmetric: valid distribution principles for conjunction and disjunction hold only in one direction.
	Individual inability does not generally compose into collective inability, while the boundary coalitions yield precise connections with systemic impossibility and strategic contingency.
	Together, these results show that inability has a stable and distinctive logical profile, despite its definability in terms of ability.
	
	The conceptual lesson is that definability does not imply dispensability.
	By shifting attention from what coalitions can enforce to what they cannot, the inability operator organises reasoning about bounded agency and negative requirements in a form that directly matches the structure of many central specifications in AI safety, protocol verification, and social choice.
	
	Future work should investigate richer notions of inability beyond the conservative, one-step setting studied here.
	The most immediate directions include epistemic inability, where lack of guaranteed enforceability is distinguished from lack of knowledge of one's abilities~\cite{HW03,AA19}; resource-bounded inability, where limitations arise from finite budgets rather than from the strategic environment~\cite{ADL14,ABL17}; and dynamic inability, where the temporal evolution of strategic limitations is modelled within frameworks such as \(\ATL\)~\cite{AHK02} or dynamic logics of games~\cite{GHL22,vDHK07}.
	The analogy with awareness and unawareness in epistemic logic suggests that inability, although introduced here as a definitional extension, may motivate independent semantic treatments in such richer settings.
	
	A particularly interesting direction is the comparative study of \emph{strategic impotence}---the condition \(\Iab_C\varphi \land \Iab_C\neg\varphi\), expressing a coalition's complete lack of control over a proposition---and what may be called \emph{absolute inability}: the grand coalition's inability \(\Iab_N\varphi\), expressing systemic impossibility.
	Despite its formal equivalence to \(\Eff{\emptyset}\neg\varphi\) (Theorem~\ref{thm:grand-coalition}), absolute inability as a coalitional modality may play a conceptually significant role in understanding the boundary between coalitional agency and systemic constraint.
	Formal reducibility does not settle conceptual significance, and these derived notions merit further investigation.
	
	The logic of inability is therefore a logic of boundaries.
	It identifies where agency stops, where constraints begin, and which outcomes remain beyond the reach of strategic power.
	For the analysis of multi-agent systems, such negative structure is not peripheral but fundamental.


\begin{thebibliography}{49}
		
		
		\bibitem{Austin56}
		J.~L. Austin.
		\newblock Ifs and cans.
		\newblock \emph{Proceedings of the British Academy}, 42:109--132, 1956.
		
		\bibitem{Kenny75}
		Anthony Kenny.
		\newblock \emph{Will, Freedom and Power}.
		\newblock Blackwell, 1975.
		
		\bibitem{Maier22}
		John Maier.
		\newblock Abilities.
		\newblock In Edward N.~Zalta, editor, \emph{The Stanford Encyclopedia of Philosophy}. Spring 2022 edition, 2022.
		
		\bibitem{Vranas07}
		Peter B.~M. Vranas.
		\newblock I ought, therefore I can.
		\newblock \emph{Philosophical Studies}, 136(2):167--216, 2007.
		
		\bibitem{Bostrom14}
		Nick Bostrom.
		\newblock \emph{Superintelligence: Paths, Dangers, Strategies}.
		\newblock Oxford University Press, 2014.
		
		\bibitem{Amodei16}
		Dario Amodei, Chris Olah, Jacob Steinhardt, Paul Christiano,
		John Schulman, and Dan Man{\'e}.
		\newblock Concrete problems in AI safety.
		\newblock \emph{arXiv preprint arXiv:1606.06565}, 2016.
		
		\bibitem{Ji23}
		Jiaming Ji, Tianyi Qiu, Boyuan Chen, et al.
		\newblock AI alignment: A comprehensive survey.
		\newblock \emph{arXiv preprint arXiv:2310.19852}, 2023.
		
		\bibitem{BH24}
		Yoshua Bengio, Geoffrey Hinton, Andrew Yao, Dawn Song, Pieter Abbeel,
		Trevor Darrell, Yuval Noah Harari, Ya-Qin Zhang, Lan Xue,
		Shai Shalev-Shwartz, Gillian Hadfield, Jeff Clune, Tegan Maharaj,
		Frank Hutter, Atılım G{\"u}ne{\c{s}} Baydin, Sheila McIlraith,
		Qiqi Gao, Ashwin Acharya, and others.
		\newblock Managing extreme AI risks amid rapid progress.
		\newblock \emph{Science}, 384(6698):842--845, 2024.
		
		\bibitem{Gibbard73}
		Allan Gibbard.
		\newblock Manipulation of voting schemes: A general result.
		\newblock \emph{Econometrica}, 41(4):587--601, 1973.
		
		\bibitem{Satterthwaite75}
		Mark Allen Satterthwaite.
		\newblock Strategy-proofness and Arrow's conditions: Existence and correspondence theorems for voting procedures and social welfare functions.
		\newblock \emph{Journal of Economic Theory}, 10(2):187--217, 1975.
		
		
		\bibitem{Pauly02}
		Marc Pauly.
		\newblock A modal logic for coalitional power in games.
		\newblock \emph{Journal of Logic and Computation}, 12(1):149--166, 2002.
		
		\bibitem{Pauly01}
		Marc Pauly.
		\newblock \emph{Logic for Social Software}.
		\newblock PhD thesis, University of Amsterdam, 2001.
		
		\bibitem{AHK02}
		Rajeev Alur, Thomas A. Henzinger, and Orna Kupferman.
		\newblock Alternating-time temporal logic.
		\newblock \emph{Journal of the ACM}, 49(5):672--713, 2002.
		
		\bibitem{Wooldridge09}
		Michael Wooldridge.
		\newblock \emph{An Introduction to MultiAgent Systems}.
		\newblock Wiley, 2nd edition, 2009.
		
		\bibitem{BGJ15}
		Nils Bulling, Valentin Goranko, and Wojciech Jamroga.
		\newblock Logics for reasoning about strategic abilities in multi-agent systems.
		\newblock In Hans van Ditmarsch, Joseph Y. Halpern, Wiebe van der Hoek,
		and Barteld Kooi, editors, \emph{Handbook of Epistemic Logic},
		pages 253--312. College Publications, 2015.
		
		\bibitem{HW03}
		Wiebe van der Hoek and Michael Wooldridge.
		\newblock Cooperation, knowledge, and time: Alternating-time temporal epistemic logic and its applications.
		\newblock \emph{Studia Logica}, 75(1):125--157, 2003.
		
		\bibitem{AA19}
		Thomas {\AA}gotnes and Natasha Alechina.
		\newblock Coalition logic with individual, distributed and common knowledge.
		\newblock \emph{Journal of Logic and Computation}, 29(7):1041--1069, 2019.
		
		\bibitem{AJ22}
		Thomas {\AA}gotnes and Wojciech Jamroga.
		\newblock Group and individual reasoning about knowledge and ability.
		\newblock \emph{Artificial Intelligence}, 310:103752, 2022.
		
		\bibitem{ADL14}
		Natasha Alechina, St{\'e}phane Demri, and Brian Logan.
		\newblock Reasoning about resource-bounded agents.
		\newblock \emph{Journal of Logic and Computation}, 24(3):661--697, 2014.
		
		\bibitem{ABL17}
		Natasha Alechina, Nils Bulling, and Brian Logan.
		\newblock On the boundary of decidability: Decidable model-checking for a fragment of resource agent logic.
		\newblock In \emph{Proceedings of IJCAI 2017}, pages 1494--1500, 2017.
		
		\bibitem{GLP24}
		Valentin Goranko, Munyque Mittelmann, and Giuseppe Perelli.
		\newblock Coalition logic with constraints on actions.
		\newblock In \emph{Proceedings of AAMAS 2024}, pages 701--709, 2024.
		
		\bibitem{LH22}
		Emiliano Lorini and Andreas Herzig.
		\newblock A logic of individual and collective agency with contingent action types.
		\newblock \emph{Artificial Intelligence}, 311:103770, 2022.
		
		\bibitem{MMPV14}
		Fabio Mogavero, Aniello Murano, Giuseppe Perelli, and Moshe Y. Vardi.
		\newblock Reasoning about strategies: On the model-checking problem.
		\newblock \emph{ACM Transactions on Computational Logic}, 15(4):34:1--34:47, 2014.
		
		\bibitem{CLMR23}
		Edoardo Caravagna, Alessio Lomuscio, Aniello Murano, and Giuseppe Perelli.
		\newblock Verification of strategy logic specifications.
		\newblock In \emph{Proceedings of IJCAI 2023}, pages 6575--6583, 2023.
		
		\bibitem{Hintikka62}
		Jaakko Hintikka.
		\newblock \emph{Knowledge and Belief: An Introduction to the Logic of the Two Notions}.
		\newblock Cornell University Press, 1962.
		
		
		\bibitem{HM85}
		Joseph Y. Halpern and Yoram Moses.
		\newblock A guide to the modal logics of knowledge and belief.
		\newblock In \emph{Proceedings of IJCAI 1985}, pages 480--490, 1985.
		
		\bibitem{FH88}
		Ronald Fagin and Joseph Y. Halpern.
		\newblock Belief, awareness, and limited reasoning.
		\newblock \emph{Artificial Intelligence}, 34(1):39--76, 1988.
		
		\bibitem{vdHL04}
		Wiebe van der Hoek and Alessio Lomuscio.
		\newblock A logic for ignorance.
		\newblock \emph{Electronic Notes in Theoretical Computer Science}, 85(2):117--133, 2004.
		
		\bibitem{BFG23}
		Stefano Bonzio, Vincenzo Fano, and Pierluigi Graziani.
		\newblock A logical modeling of severe ignorance.
		\newblock \emph{Journal of Philosophical Logic}, 52(4):1053--1080, 2023.
		
		\bibitem{FH24}
		Hana Frluckaj and Eric Pacuit.
		\newblock The logic of ignorance: A proof-theoretic perspective.
		\newblock \emph{Studia Logica}, 112(2):341--372, 2024.
		
		\bibitem{DLR98}
		Eddie Dekel, Barton L. Lipman, and Aldo Rustichini.
		\newblock Standard state-space models preclude unawareness.
		\newblock \emph{Econometrica}, 66(1):159--173, 1998.
		
		\bibitem{Schipper15}
		Burkhard C. Schipper.
		\newblock Awareness.
		\newblock In Hans van Ditmarsch, Joseph Y. Halpern, Wiebe van der Hoek,
		and Barteld Kooi, editors, \emph{Handbook of Epistemic Logic},
		pages 147--203. College Publications, 2015.
		
		\bibitem{DFP24}
		Franz Dietrich, Christian List, and Marcus Pivato.
		\newblock Awareness logic: A Kripke-based rendition.
		\newblock \emph{Journal of Philosophical Logic}, 53:1--35, 2024.
		
		\bibitem{HR14}
		Joseph Y. Halpern and Leandro C. R\^{e}go.
		\newblock Reasoning about knowledge of unawareness.
		\newblock \emph{Games and Economic Behavior}, 88:100--120, 2014.
		
		\bibitem{JLN23}
		Wojciech Jamroga, Damian Le{\'s}kiewicz, and Artur Niewiadomski.
		\newblock Model checking strategic ability under imperfect information is undecidable.
		\newblock \emph{Journal of Artificial Intelligence Research}, 76:1--35, 2023.
		
		\bibitem{GHL22}
		Valentin Goranko, Antti Kuusisto, and Raine R{\"o}nnholm.
		\newblock Game-theoretic semantics for ATL+ with applications to model checking.
		\newblock In \emph{Proceedings of AAMAS 2022}, pages 559--567, 2022.
		
		
		\bibitem{GJ04}
		Valentin Goranko and Wojciech Jamroga.
		\newblock Comparing semantics of logics for multi-agent systems.
		\newblock \emph{Synthese}, 139(2):241--280, 2004.
		
		\bibitem{Chellas80}
		Brian F. Chellas.
		\newblock \emph{Modal Logic: An Introduction}.
		\newblock Cambridge University Press, 1980.
		
		\bibitem{BRV01}
		Patrick Blackburn, Maarten de Rijke, and Yde Venema.
		\newblock \emph{Modal Logic}.
		\newblock Cambridge University Press, 2001.
		
		\bibitem{Enderton01}
		Herbert B. Enderton.
		\newblock \emph{A Mathematical Introduction to Logic}.
		\newblock Academic Press, 2nd edition, 2001.
		
		
		\bibitem{OR94}
		Martin J. Osborne and Ariel Rubinstein.
		\newblock \emph{A Course in Game Theory}.
		\newblock MIT Press, 1994.
		
		\bibitem{BPX01}
		Nuel Belnap, Michael Perloff, and Ming Xu.
		\newblock \emph{Facing the Future: Agents and Choices in Our Indeterminist World}.
		\newblock Oxford University Press, 2001.
		
		\bibitem{Horty01}
		John F. Horty.
		\newblock \emph{Agency and Deontic Logic}.
		\newblock Oxford University Press, 2001.
		
		\bibitem{BH15}
		Jan Broersen and Andreas Herzig.
		\newblock Using STIT theory to talk about strategies.
		\newblock In Johan van Benthem, Sujata Ghosh, and Rineke Verbrugge, editors,
		\emph{Models of Strategic Reasoning}, pages 137--173. Springer, 2015.
		
		
		\bibitem{FHMV95}
		Ronald Fagin, Joseph Y. Halpern, Yoram Moses, and Moshe Y. Vardi.
		\newblock \emph{Reasoning About Knowledge}.
		\newblock MIT Press, 1995.
		
		
		\bibitem{Arrow51}
		Kenneth J. Arrow.
		\newblock \emph{Social Choice and Individual Values}.
		\newblock Wiley, 1951.
		
		\bibitem{Sen70}
		Amartya Sen.
		\newblock \emph{Collective Choice and Social Welfare}.
		\newblock Holden-Day, 1970.
		
		\bibitem{BCG16}
		Felix Brandt, Vincent Conitzer, Ulle Endriss, J{\'e}r{\^o}me Lang,
		and Ariel D. Procaccia, editors.
		\newblock \emph{Handbook of Computational Social Choice}.
		\newblock Cambridge University Press, 2016.
		
		\bibitem{vDHK07}
		Hans van Ditmarsch, Wiebe van der Hoek, and Barteld Kooi.
		\newblock \emph{Dynamic Epistemic Logic}.
		\newblock Cambridge University Press, 2007.
		
	\end{thebibliography}
\end{document}